\documentclass{JUSTC}

\RequirePackage[numbers,sort&compress]{natbib} 
\usepackage{url}
\usepackage{hyperref}

\usepackage{algorithm,algorithmicx,algpseudocode}%
\algrenewcommand\algorithmicrequire{\textbf{Input:}}  
\algrenewcommand\algorithmicensure{\textbf{Output:}} 

\type{Research Articles}
\title{Robust Wasserstein barycenter}

\author[1]{Zixiong Cheng}

\author[1,\Letter]{Hang Liu}

\email{hliu01@ustc.edu.cn}

\runningtitle{Robust Wasserstein barycenter}

\address{School of Management, University of Science and Technology of China, Jinzhai Rd, Hefei, 230026, Anhui Province}

\GraphicalAbstract{
\begin{figure}[htb]
	\centering     
	\includegraphics[width=5in]{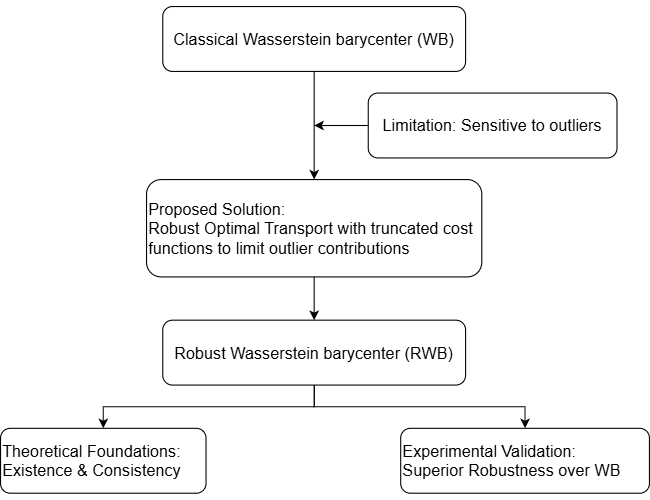}
\end{figure}
}

\abstract{
In this paper, we address a fundamental limitation of the classical Wasserstein barycenter—its sensitivity to outliers. 
To overcome these issues, we propose the robust Wasserstein barycenter (RWB) based on the recent concept of the robust optimal transport. Theoretical guarantees, including existence and consistency,
are established for the proposed RWB. Through extensive numerical experiments on both simulated and real-world data---including image processing and financial data analysis---we demonstrate that compared with the classical Wasserstein barycenter, the RWB is more robust.
}


\begin{document}

\maketitle

\section{Introduction}

The theory of optimal transport (OT), first introduced by \citet{monge1781} in the late 18th century, was revisited and generalized by \citet{Kantorovich2006OnTT} nearly two centuries later. Since the advent of efficient computational methods, OT has been broadly applied across various disciplines, including machine learning and statistics. 
In machine learning, \citet{WGAN} introduced WGAN as a GAN alternative, using the Wasserstein distance to enhance training stability, which avoids mode collapse, and reveals deep connections to other distribution distances. 
In statistics, the Wasserstein distance is used to quantify the discrepancy between probability distributions. Particularly in compressive statistical learning (CSL), \citet{2023CSL} provide a theoretical guarantee for the framework of CSL by linking the Wasserstein distance to the maximum mean discrepancies (MMD).
\citet{Srivastava} applied the Wasserstein distance and barycenter in divide-and-conquer Bayesian inference, using barycenters of subset posterior distributions in Wasserstein space to efficiently approximate full-data posteriors with theoretical accuracy guarantees.
\citet{Ramdas2015OnWT} took the Wasserstein distance as the core in nonparametric two-sample testing to connect diverse univariate and multivariate testing methods via smoothed or distribution-free variants.

The Wasserstein barycenter is defined as the distribution that minimizes the sum of Wasserstein distances to a given set of distributions, and it offers an approach to averaging multiple distributions, generalizing the concept of the barycenter to the space of probability measures.

\citet{2017Existence} established the existence and consistency of the $p$-Wasserstein barycenter, thereby providing a rigorous theoretical foundation for its application. Moreover, \citet{2013Fast} demonstrated the computational efficiency and practical utility of the Wasserstein barycenter, yielding remarkable applications such as image recognition and clustering.
Despite its wide applications in statistics, economics, and machine learning, the Wasserstein barycenter has a critical limitation: it is sensitive to outliers and contamination. 
In conventional dataset-related problems, robustness is typically achieved by removing extreme values; however, the Wasserstein barycenter, which is defined as the center of mass with respect to a collection of probability measures, complicates the identification of outliers. 

In \cite{ma2023theoretical}, the robust Wasserstein distance is defined by the concept of robust optimal transport (ROT) in \cite{ROT}, which truncates the distance function integrated in the computation of the Wasserstein distance, thus ensuring the robustness of the distance measure. 
Notably, this framework considers only the robust version of the classical $W_p$ distance for $p = 1$. 
We generalize the robust Wasserstein distance to the case of \( p \geq 1 \) and introduce the robust Wasserstein barycenter (RWB) as a robust alternative to the classical Wasserstein barycenter (WB).  
Extensive experiments in Sections 5 and 6 demonstrate that the RWB is more robust than the WB. 
With respect to the computation of the WB, \citet{simple_free} noted that although the iterative Bregman projections (IBP) algorithm \cite{IBP} is highly efficient and straightforward in most cases, it is a fixed-support algorithm implemented on a fixed predefined grid support, thus it has the curse of dimensionality. 
In contrast, free-support algorithms eliminate such grid-based support restrictions. Since the optimal barycenter solution is naturally sparse, the free-support framework effectively alleviates the curse of dimensionality.
Therefore, we utilize the simple and well-known fixed-point algorithm mentioned in \citet{simple_free} to improve the performance of the fixed-support IBP algorithm for our free-support method in the experiments. 
We find that, in addition to its intrinsic capacity to circumvent the curse of dimensionality, the free-support method yields superior barycenter results in our image processing experiments.

The structure of the paper is organized as follows.
Section 2 describes the essential concepts related to the robust Wasserstein barycenter, including optimal transport, the Fr{\'e}chet mean, and the robust Wasserstein distance. 
Section 3 proposes the new concept of the RWB and proves the existence and consistency of the RWB. 
Section 4 introduces a free-support algorithm for computing the RWB, which is built on existing fixed-support algorithms. 
Finally, Section 5 compares the RWB, the Wasserstein median proposed in \cite{YouShun} and WB through simulation and real-data analysis, illustrating the advantages of the RWB.

The main contributions of this paper are threefold: 
\begin{enumerate}
	\item[(i)] We propose the novel concept of the robust Wasserstein barycenter which overcomes the robustness issue of the classical Wasserstein barycenter.
	\item[(ii)] We prove the existence and consistency of the RWB. Furthermore, we characterize the size of the support of the robust Wasserstein barycenter under the condition that all input distributions are discrete with finite supports.
	\item[(iii)] We introduce a free-support method for computing RWB, which strengthens existing fixed-support algorithms and resolves the difficulty of selecting proper supports for barycenter computation encountered by fixed-support approaches.
\end{enumerate}

\section{Preliminaries}

In this section, we introduce the concepts of the Fr{\'e}chet mean, optimal transport and Wasserstein distance that are relevant for the definition of the robust Wasserstein barycenter.

\subsection{Fr{\'e}chet mean}
Given a metric space $(\mathcal{X}, d)$ and a positive integer $k$, {\it $k$-Fr{\'e}chet mean} introduced in \cite{Frechet1948} is defined for a random variable $X \in \mathcal{X}$ as
$$\theta_k := \underset{x \in \mathcal{X}}{\arg \min} \mathbb{E}[d^{\mathrm{k}}\left(X,x\right)].$$
The $k$-Fr{\'e}chet mean generalizes the concept of the mean from Euclidean space to generic metric spaces. When $(\mathcal{X}, d)$ is Euclidean, the \(k\)-Fr{\'e}chet mean coincides with some well-known measures of central tendency: for \( \mathrm{k} = 1 \), it corresponds to the usual median, and for \(k = 2\), it corresponds to the arithmetic mean.

Given observations $X_1, . . . , X_n$ that are i.i.d. copies of $X$, the {\it empirical $k$-Fr{\'e}chet mean} is naturally defined as: 
$$\hat{\theta}_k := \underset{x \in \mathcal{X}}{\arg \min} \sum_{i=1}^{n} d^\mathrm{k}\left(X_i, x\right). $$

\subsection{Optimal transport}
Let $({\cal X}, d)$ and $({\cal Y}, d)$ be two complete and separable metric spaces, and let ${\cal P}({\cal X})$ denote the set of all probability measures on ${\cal X}$. Given two probability measures $\mu \in \mathcal{P(X)}$, $\nu \in \mathcal{P(Y)}$ and a cost function $c: \mathcal{X \times Y}\rightarrow [0,+\infty]$, the optimal transport problem first proposed by \cite{monge1781} is to solve the following minimization equation
$$\inf \left\lbrace  {\int_{\mathcal{X}} c(x, \mathcal{T}(x)) d \mu(x): \mathcal{T}\# \mu =\nu}  \right\rbrace,$$
where $\mathcal{T} \# \mu =\nu$ indicates that $\mathcal{T}$ pushes $\mu$ forward to $\nu$. While Monge's concept of finding a push-forward map that minimizes the total transport cost is quite intuitive, it faces two critical limitations that hinder its practical implementation. 
First, solutions may not exist for certain measures. For example, no map can push a discrete probability measure forward to a continuous one. Second, the set $\{\mu, {\cal T}(\nu)\}$ of all measurable transport maps is nonconvex, which makes Monge's minimization problem computationally challenging.

Monge’s problem was revisited by \citet{Kantorovich2006OnTT}, who proposed a more flexible and computationally feasible formulation. Let $\Pi (\mu,\nu)$ denote the set of all couplings of $\mu \in \mathcal{P(X)}$ and $\nu \in \mathcal{P(Y)}$. Kantorovich's problem aims to find an optimal coupling $\gamma   \in \Pi (\mu,\nu)$ between $X$ and $Y$, which is formulated as:
$$\inf \left\lbrace \int_{\mathcal{X\times Y}} c(x,y) d \gamma(x,y) : \gamma \in \Pi(\mu,\nu) \right\rbrace.$$

Let $\mathcal{P}_{p}(\mathcal{X}) = \{\mu \in \mathcal{P}(\mathcal{X}): \int_{\mathcal{X}} d^{p}(x, x_0)  \mathrm{d} \mu<+\infty \, \text{for some} \, x_0 \in {\cal X}\}$ denote the set of all probability measures on ${\cal X}$ with a finite moment of order $p$ $(p \in [1,+\infty))$. Given two probability measures $\mu, \nu \in \mathcal{P}_{p}(\mathcal{X})$, the $p$-Wasserstein distance $W_p(\mu, \nu)$ is defined as the minimum value of the optimal transport problem, that is,
$$W_{p}\left(\mu, \nu\right):=\underset{\gamma \in  \Pi\left(\mu,\nu\right)}{\inf} \left\{\int_{\mathcal{X} \times \mathcal{X}}d^{p}\left(x,y\right) \mathrm{d} \gamma \right\}^{\frac{1}{p}}.$$ 
In contrast to the Kullback-Leibler (KL) divergence \cite{KL} and the Jensen-Shannon (JS) divergence \cite{JS}, the Wasserstein distance is a metric between probability measures. 
Consequently, $\mathcal{P}_{p}$, when endowed with $W_p$, is a metric space referred to as the Wasserstein space \cite{villani2009}. 
Another advantage of the Wasserstein distance over the KL divergence is that the latter is well defined only if two probability measures have the same support, whereas the former can be applied to probability measures with different supports.

\section{Robust Wasserstein barycenter}
In this section, building upon the robust optimal transport in \cite{ROT} and \cite{ma2023theoretical}, we introduce the robust Wasserstein distance $W_p^{(\lambda)}$ as a robust version of the classical Wasserstein distance $W_p$. Building upon the $W_p^{(\lambda)}$, we propose the concept of the robust Wasserstein barycenter (RWB) and establish the existence and consistency of the RWB.

\subsection{ Robust Wasserstein distance}
Since the Wasserstein distance is sensitive to outliers, the main idea of the robust optimal transport frameworks in \cite{ROT} and  \cite{ma2023theoretical} is to truncate the cost function of the classical $W_1$ distance, thereby yielding the minimization problem:
$$\underset{\gamma \in  \Pi\left(\mu,\nu\right)}{\inf} \left\{\int_{\mathcal{X}\times \mathcal{X}}c^{\left(\lambda\right)}\left(x,y\right) \mathrm{d} \gamma \right\},
$$
where
\begin{equation}\label{eq.c_lam}
	c^{(\lambda)}(x, y) := \min\{d(x,y),\lambda\},
\end{equation} 
with $\lambda>0$ as a regularization parameter controlling the truncation level. 
According to \cite{ma2023theoretical}, since $c^{(\lambda)}$ is a metric on $\mathcal{X}$, the robust Wasserstein distance is defined as:
$$W^{(\lambda)}(\mu, \nu):=\underset{\gamma \in  \Pi(\mu, \nu)}{\min} \left\{\int_{\mathcal{X} \times \mathcal{X}}c^{(\lambda)}(x, y) \mathrm{d} \gamma \right\},$$
which is a metric on $\mathcal{P(X)}$ and coincides with the classical $W_1$ distance on the space of probability measures over the metric space $({\cal X}, c^{(\lambda)})$.

We consider a $p$-th generalized version of $W^{(\lambda)}$ and define the {\it  robust $W_p$ distance} ($p\geq 1$) as
$$W_p^{(\lambda)}(\mu, \nu) := \underset{\gamma \in  \Pi(\mu, \nu)}{\min} \left\{\int_{\mathcal{X} \times \mathcal{X}} \left[c^{(\lambda)}(x, y)\right]^{p} \mathrm{d} \gamma \right\}^{\frac{1}{p}}.
$$
Like $W^{(\lambda)}$, $W_p^{(\lambda)}$ is a metric on $\mathcal{P(X)}$. We refer to $\mathcal{P(X)}$ endowed with the metric $W_p^{(\lambda)}$ as the {\it $p$-th robust Wasserstein space}, which we denote as ${\cal W}_p^{(\lambda)}({\cal X})$. 
The most commonly adopted values for \( p \) are \( p = 1 \) and \( p = 2 \). As noted in \cite{villani2009}, the \( W_1 \) Wasserstein distance is more robust and less sensitive to outliers, whereas the \( W_2 \) Wasserstein distance exhibits superior geometric properties and is more computationally efficient. 
Since the robust Wasserstein distance remains within the framework of classical Wasserstein distance, we conjecture that $W_p^{(\lambda)}$ inherits these respective characteristics.

\subsection{Robust Wasserstein barycenter}
\textcolor{black}{We now propose the robust Wasserstein barycenter.  Consider a random probability measure $\mu$ in ${\cal W}_p^{(\lambda)}({\cal X})$ that follows a distribution $\Lambda$, where $\Lambda$ lies in the space ${\cal W}_k({\cal W}_p^{(\lambda)}({\cal X}))$ ($k\geq1$), which is meaningful because ${\cal W}_p^{(\lambda)}({\cal X})$ is a complete and separable metric space \cite{villani2009}. To distinguish the notation, we denote $\mathbb{W}_k$ as the metric defined on $\mathcal{W}_k(\mathcal{W}_p^{(\lambda)}(\mathcal{X}))$. Furthermore, since the robust Wasserstein distance satisfies $W_p^{(\lambda)} \leq \lambda$ identically, $\mathbb{W}_k^{(\lambda)}$ is equivalent to $\mathbb{W}_k$. This further indicates that the two spaces $\mathcal{W}_k^{(\lambda)}(\mathcal{W}_p^{(\lambda)}(\mathcal{X}))$ and $\mathcal{W}_k(\mathcal{W}_p^{(\lambda)}(\mathcal{X}))$ are identical.
Then for any $\nu \in {\cal W}_p^{(\lambda)}({\cal X})$, we can write: 
$$ 
[\mathbb{W}_k(\Lambda, \delta_{\nu})]^k = \int [W_p^{(\lambda)}(\mu, \nu)]^k d\Lambda(\mu) = \mathbb{E}[W_p^{(\lambda)}(\mu, \nu)]^k,
$$
}
where $\delta_{\nu}$ denotes a Dirac measure.

\begin{definition}{\rm 
        \textcolor{black}{For a random measure $\mu$ with distribution  $\Lambda \in {\cal W}_k({\cal W}_p^{(\lambda)}({\cal X}))$, 
        the robust Wasserstein barycenter $\nu_{p, k}^{(\lambda)}$ $(k \geq 1)$ is defined as:
		$$\nu_{p, k}^{(\lambda)} :=\underset{\xi \in {\cal W}_p^{(\lambda)}({\cal X})}{\arg\min}\mathbb{E}\left[W_p^{(\lambda)}(\mu, \xi)\right]^k.$$}
	}
\end{definition}\label{def.RWB}

Let $\mu_i \in {\cal W}_p^{(\lambda)}({\cal X})$, $i = 1, \ldots, n$,  be a sequence of probability measures. The RWB with respect to $\Lambda =\sum_{i=1}^{n}w_i\delta_{\mu_i}$ can be written as:
$${\nu}_{p, k}^{(\lambda, n)} :=  \underset{\xi \in {\cal W}_p^{(\lambda)}({\cal X})} {\arg\min} \sum_{i=1}^{n}w_i \left[W_p^{(\lambda)}\left({\mu_i}, \xi\right)\right]^k,$$
where $\{\delta_{\mu_i}\}_{i=1}^{n}$ denotes a sequence of Dirac measures, each concentrated at a point \(\mu_i\) in ${\cal W}_p^{(\lambda)}({\cal X})$. The weights $w_i$ for $i = 1, \ldots, n$, satisfy \(\sum_{i=1}^{n}w_i=1\), ensuring that $\sum_{i=1}^{n}w_i\delta_{\mu_i}$ is a valid distribution. 
Additionally, in applications, it is typical to set \( k = p \), which is the most common form of the Wasserstein barycenter and satisfies the requirements for most scenarios.

The following result establishes the existence of the RWB $\nu_{p,k}^{(\lambda)}$ for any distribution $\Lambda \in {\cal W}_k({\cal W}_p^{(\lambda)}({\cal X}))$.

\begin{proposition}[Existence]\label{prop.exist}
	\textcolor{black}{Let \((\mathcal{X}, d)\) be a complete, separable, and compact metric space. Let $\mu$ be a random measure with distribution $\Lambda \in {\cal W}_k({\cal W}_p^{(\lambda)}({\cal X}))$. For any $k \ge 1$, the functional
	$ F(\nu) = \mathbb{E}\left[W_p^{(\lambda)}(\nu, \mu)\right]^k $
	admits a minimizer $\nu^* \in \mathcal{W}_p^{(\lambda)}(\mathcal{X})$.}
\end{proposition}

\begin{proof}
	Obviously, for any \(\nu \in {\cal W}_p^{(\lambda)}(\mathcal{X})\), we have: $$F(\nu) = \mathbb{E}\left[ W_{p}^{(\lambda)}(\nu, \textcolor{black}{\mu})\right]^k \leq \lambda < \infty.$$ We can take a minimizing sequence $\left\{\nu_j\right\} \subset {\cal W}_p^{(\lambda)}(\mathcal{X})$.
	For any $R>0$ and any fixed point \( x_0 \in \mathcal{X}\), we have that
	$$
	\nu_j\left(\left\{x \in \mathcal{X} \mid {c^{\left(\lambda\right)}\left(x,x_0\right)} \geq R\right\}\right) \leq \frac{1}{R} \int_{\{x:{c^{\left(\lambda\right)}\left(x,x_0\right)} \geq R\}}{c^{\left(\lambda\right)}\left(x,x_0\right)} d \nu_j\left(x\right) \leq \frac{\lambda}{R}.
	$$
	For any \(\epsilon > 0 \), we can choose a sufficiently large \( R \) such that \( R > \lambda/\epsilon \). In fact, we might let \( \epsilon < 1 \), then \( R > \lambda \).
	For the closed ball $$B_R(x_0)=\left\{x \in \mathcal{X} \mid {c^{\left(\lambda\right)}\left(x,x_0\right)} \leq R\right\} = \mathcal{X},$$ we have
	$\nu_j(B_R) = 1 > 1 - \epsilon,$
	for any \( j \).
	
	{The sequence of measures $\left\{\nu_j\right\}$ is tight since \( ({\cal X}, c^{\lambda}) \) is a complete, separable and compact metric space.} 
	By the Prokhorov’s theorem {\cite{Prokhorov}}, there exists a subsequence $\left\{\nu_{j_l}\right\}$ that converges weakly to some $\nu^* \in {\cal W}_p^{(\lambda)}(\mathcal{X})$.  
	Then we have 
	$$
	\begin{aligned}
		F\left(\nu^*\right)
		& =\mathbb{E}\left[W_{p}^{\left(\lambda\right)}\left(\nu^*, \textcolor{black}{\mu}\right)\right]^k \\ 
		& \leq
		\mathbb{E} \liminf _{j_l \rightarrow \infty} \left[ W_{p}^{\left(\lambda\right)} \left(\nu_{j_l}, \textcolor{black}{\mu}\right)\right]^k \\
		& \leq 
		\liminf _{j_l \rightarrow \infty} \mathbb{E}\left[ W_{p}^{\left(\lambda\right)}\left(\nu_{j_l}, \textcolor{black}{\mu}\right)\right]^k \\
		&=\inf F,
	\end{aligned}
	$$ 
	where the first inequality is due to the lower semicontinuity of the Wasserstein distance \cite{Panaretos2020}, and the second inequality is a consequence of Fatou's lemma. This completes the proof. 
\end{proof}

Moreover, we also establish the consistency of the RWB, which is stated as follows.
\begin{proposition}[Consistency]\label{prop.consistency}
	Let \textcolor{black}{$\mu_j \sim \Lambda_j $}, $j \in \mathbb{N}$ be a sequence of random measures with a sequence of robust Wasserstein barycenters  $\nu_{p, k, j}^{(\lambda)}$ ($k, p \geq 1$). Assuming that $\textcolor{black}{\mathbb{W}_k}(\Lambda_j, \Lambda) \rightarrow 0$ \textcolor{black}{holds for a distribution $\Lambda \in \mathcal{W}_k(\mathcal{W}_p^{(\lambda)}(\mathcal{X}))$ as $j \rightarrow \infty$.} Then the sequence $\{\nu_{p, k, j}^{(\lambda)}\}$ is tight and \textcolor{black}{every weak subsequential limit} is a robust Wasserstein barycenter of $\Lambda$.  
\end{proposition}
\begin{proof}
According to the proof of proposition~\ref{prop.exist}, the sequence of robust Wasserstein barycenters $\{\nu_{p, k, j}^{(\lambda)}\}$ is tight. We can take any subsequence $\{\nu_{p, k, j_l}^{(\lambda)}\}$ that weakly converges to some $\nu^* \in {\cal W}_p^{(\lambda)}(\mathcal{X})$. 

Next, we will prove that the point $\nu^*$ is also a robust Wasserstein barycenter. 
Considering random measures $\mu_{j_l}$ and $\mu$, which are random copies of $\Lambda_{j_l}$ and $\Lambda$, respectively. 
Then, we aim to prove that $\mathbb{E}\left[ W_{p}^{\left(\lambda\right)}\left(\xi, \mu\right)\right]^k \geq \mathbb{E}\left[ W_{p}^{\left(\lambda\right)}\left(\nu^*, \mu\right)\right]^k$ holds for any choice of $\xi \in$ $\mathcal{W}_p^{\left(\lambda\right)}\left(\mathcal{X}\right)$.
\begin{equation}  
	\begin{aligned}
		\mathbb{E} \left[ W_{p}^{\left(\lambda\right)}\left(\xi, \mu\right)\right]^k 
		& = \textcolor{black}{\left[\mathbb{W}_k\left(\delta_{\xi}, \Lambda\right)\right]^k} \\
		& = \textcolor{black}{\lim _{j_l \rightarrow \infty} \left[\mathbb{W}_k\left(\delta_{\xi}, \Lambda_{j_l}\right)\right]^k \text{since $\mathbb{W}_k(\Lambda_j, \Lambda) \rightarrow 0$} }\\
		& =\lim _{j_l \rightarrow \infty} \mathbb{E} \left[W_{p}^{\left(\lambda\right)}\left(\xi, \mu_{j_l}\right)\right]^k   \\
		& \geq \lim _{j_l \rightarrow \infty} \mathbb{E} \left[W_{p}^{\left(\lambda\right)}\left(\nu_{p, k, j_l}^{(\lambda)}, \mu_{j_l}\right)\right]^k \\
		& \geq \mathbb{E} \liminf _{j_l \rightarrow \infty} \left[W_{p}^{\left(\lambda\right)}\left(\nu_{p, k, j_l}^{(\lambda)}, \mu_{j_l}\right)\right]^k \text{ by Fatou's lemma } \\
		& \geq \mathbb{E} \left[W_{p}^{\left(\lambda\right)}\left(\nu^*, \mu\right)\right]^k.
	\end{aligned}
	\label{proof2}
\end{equation} 
The first inequality holds because every element in the sequence $\left\{\nu_{p, k, j_l}^{(\lambda)}\right\}$ is a robust Wasserstein barycenter that minimizes the expected robust Wasserstein distance to the random measure \textcolor{black}{\(\mu_{j_l} \sim \Lambda_{j_l}\)}. 
Furthermore, the convergence of $\Lambda_{j_l} \rightarrow \Lambda$ allows us to construct $\mu_{j_l} \rightarrow \mu$ almost surely by Skorokhod's representation theorem in \cite{1999Convergence}. 
Then the last inequality holds because the Wasserstein distance is lower semicontinuous and the robust Wasserstein distance is a special case of the Wasserstein distance.

We then aim to further prove $W_{p}^{\left(\lambda\right)}\left(\nu_{p, k, j_l}^{(\lambda)},\nu^*\right) \rightarrow 0$. 
If $\xi=\nu^*$, Equation \eqref{proof2} is in fact an equality, which implies that \textcolor{black}{$ \mathbb{W}_k\left(\delta_{\nu_{p, k, j_l}^{(\lambda)}}, \Lambda_{j_l}\right) \rightarrow \mathbb{W}_k\left(\delta_{\nu^*}, \Lambda\right)$}. 
By the triangle inequality, we have
\begin{align*}
& \textcolor{black}{\mathbb{W}_k\left(\delta_{\nu_{p, k, j_l}^{(\lambda)}}, \Lambda\right)-\mathbb{W}_k\left(\delta_{\nu^*}, \Lambda\right)}  \\
\textcolor{black}{\leq} & \textcolor{black}{\mathbb{W}_k\left(\delta_{\nu_{p, k, j_l}^{(\lambda)}}, \Lambda_{j_l}\right)+\mathbb{W}_k\left(\Lambda_{j_l}, \Lambda\right)-\mathbb{W}_k\left(\delta_{\nu^*}, \Lambda\right) \rightarrow 0}.
\end{align*}

Since $\nu^*$ is a minimizer of \textcolor{black}{$\mathbb{E}[W_{p}^{\left(\lambda\right)}\left(\nu,\mu\right)]^k$}, we have \textcolor{black}{$\left[\mathbb{W}_k\left(\delta_{\nu_{p, k, j_l}^{(\lambda)}},\Lambda\right)\right]^k
\geq
\left[\mathbb{W}_k\left(\delta_{\nu^*}, \Lambda\right)\right]^k$},
which implies that $\mathbb{W}_k\left(\delta_{\nu_{p, k, j_l}^{(\lambda)}}, \Lambda\right)-\mathbb{W}_k\left(\delta_{\nu^*}, \Lambda\right) \geq 0$. \textcolor{black}{Hence, $$\lim _{j_l \rightarrow \infty} \mathbb{W}_k\left(\delta_{\nu_{p, k, j_l}^{(\lambda)}}, \Lambda\right) = \mathbb{W}_k\left(\delta_{\nu^*}, \Lambda\right).$$
Therefore, we have
\begin{align*}
	\textcolor{black}{\mathbb{E} \left[ W_{p}^{\left(\lambda\right)}\left(\nu^*, \mu\right)\right]^k} 
	& \textcolor{black}{= \left[\mathbb{W}_k\left(\delta_{\nu^*}, \Lambda\right)\right]^k} \\
	& \textcolor{black}{= \lim _{j_l \rightarrow \infty} \left[\mathbb{W}_k\left(\delta_{\nu_{p, k, j_l}^{(\lambda)}}, \Lambda\right)\right]^k} \\
	& \textcolor{black}{=\lim _{j_l \rightarrow \infty} \mathbb{E} \left[W_{p}^{\left(\lambda\right)}\left(\nu_{p, k, j_l}^{(\lambda)}, \mu\right)\right]^k}  \\
	& \textcolor{black}{\geq \mathbb{E} \liminf _{j_l \rightarrow \infty} \left[W_{p}^{\left(\lambda\right)}\left(\nu_{p, k, j_l}^{(\lambda)}, \mu\right)\right]^k \text{ by Fatou's lemma}.}
\end{align*}
}
\textcolor{black}{Note that $\liminf\limits_{j_l \rightarrow \infty}W_{p}^{\left(\lambda\right)}\left(\nu_{p, k, j_l}^{(\lambda)}, \mu\right) \geq W_{p}^{\left(\lambda\right)}\left(\nu^*, \mu\right)$ by the lower semicontinuity of $W_{p}^{\left(\lambda\right)}$. Therefore, we have $\Lambda-a.s.$  }
$$
\textcolor{black}{\liminf _{j_l \rightarrow \infty} W_{p}^{\left(\lambda\right)}\left(\nu_{p, k, j_l}^{(\lambda)}, \mu\right)=W_{p}^{\left(\lambda\right)}\left(\nu^*, \mu\right)},
$$
which implies that $W_{p}^{\left(\lambda\right)}\left(\nu_{p, k, j_l}^{(\lambda)}, \xi\right) \rightarrow W_{p}^{\left(\lambda\right)}\left(\nu^*, \xi\right)$ for any $\xi \in {\cal W}_p^{(\lambda)}(\mathcal{X})$ almost surely with respect to $\Lambda$. By the Theorem 2.2.1 in \cite{Panaretos2020}, it is equivalent to $W_{p}^{\left(\lambda\right)}\left(\nu_{p, k, j_l}^{(\lambda)}, \nu^*\right) \rightarrow 0$, which completes the proof.

Hence, \textcolor{black}{every weak subsequential limit} of $\left\{\nu_{p, k, j}^{(\lambda)}\right\}$ is indeed a robust Wasserstein barycenter of $\Lambda$.
\end{proof}

\begin{remark}
	\label{cor:barycenter-convergence}
	If $\Lambda$ is approximated by a growing discrete distribution $\Lambda^{(j)}=\sum_{i=1}^{j} w^{(j)}_i \delta_{\mu^{(j)}_i}$, then the sequence of barycenters $\{\nu_{p,k,j}^{(\lambda)}\}$ is tight, and \textcolor{black}{every weak subsequential limit} is a robust Wasserstein barycenter of $\Lambda$.
\end{remark}

\begin{remark}
	\citet{sinhoW2_unique} characterized the uniqueness of Wasserstein barycenters via geometric conditions and optimal transport theory [Sections 8.2–8.3], and further identify the strict convexity of \textcolor{black}{\(\mathbb{E}\left[ W_p(\nu, \mu) \right]^k\)} as the core reason for the uniqueness of Wasserstein barycenters. 
	However, due to the truncation parameter \(\lambda\), the objective function \textcolor{black}{\(F(\nu) = \mathbb{E}\left[ W_p^{(\lambda)}(\nu, \mu) \right]^k\)} fails to be strictly convex. 
	Consequently, the robust Wasserstein barycenter, which can be viewed as a specific variant of the classical Wasserstein barycenter, may not be unique in general cases. Since the conditions for uniqueness are technically challenging, we leave them for the future research.
\end{remark}

\begin{remark}
	As mentioned in Section 3.1, $c^{(\lambda)}$ is verified to be a metric on $\mathcal{X}$ \cite{ma2023theoretical}. Accordingly, the robust Wasserstein distance $W_p^{(\lambda)}$, analogous to the classical Wasserstein distance, serves as a well-defined metric on $\mathcal{P}(\mathcal{X})$. Nevertheless, the alternative RWB barycenter constructed in \cite{wang2024robust} follows the outlier-robust OT paradigm introduced by \citet{nietert}. As explicitly stated in \cite{nietert}, such a robust transport kernel no longer fulfills the rigorous metric axioms on $\mathcal{P}(\mathcal{X})$.
\end{remark}

\section{Computation algorithms}\label{sec.alg}

The Wasserstein barycenter is an effective tool for summarizing collections of probability distributions and has wide applications in machine learning and statistics. 
However, computing the Wasserstein barycenter is quite challenging, as it involves solving a high-dimensional linear optimization problem.  \citet{nphard} characterized the curse of dimensionality for the Wasserstein barycenter and prove that its computation is NP-hard.

The majority of algorithms constrain the support of the barycenter to a predetermined set of grid points to compute the optimal probability mass distribution and such methods are referred to as {\it fixed-support} algorithms.
Notably, the iterative Bregman projections (IBP) algorithm proposed in \cite{IBP} is among the most widely used fixed-support algorithms.
While the IBP algorithm exhibits remarkable efficiency and simplicity in low-dimensional cases, it suffers from the curse of dimensionality because of its reliance on a predefined fixed grid as the support.
Specifically, the number of support points for such a grid increases exponentially as the dimension increases.
Although the support of the Wasserstein barycenter is naturally sparse under the assumption that the number of input distributions is finite and all the input distributions have finite support (see e.g., \citet{MFT19} and \citet{GA27}), it is difficult for fixed-support methods to identify this sparse support a priori.

Fortunately, \textit{free-support algorithms} optimize for both the probability mass weights and the support of the barycenter.
This enables the number of support points to be much smaller than that of the full grid, thus effectively alleviating the curse of dimensionality in high-dimensional cases (see e.g. \citet{2013Fast} and \citet{simple_free}).
Additionally, we find that the free-support variant of IBP performs better than the original IBP in our image processing experiments.

To this end, we first prove that the RWB has finite support when the input distributions has finite support. We then introduce a free-support algorithm for the RWB, which adopts the well-known fixed-point algorithm mentioned in \citet{simple_free}.

\subsection{Support of RWB}\label{sec.support}

In the following, $\{\mu_i\}_{i=1}^{n}$ and $\nu$ denote some finite discrete probability measures in $\mathcal{W}^{\left(\lambda\right)}_{p}(\mathcal{X})$.
We then assume that each $\mu_i$ is supported on $\text{supp}\left(\mu_i\right) = \{ x_1^{\left(i\right)}, \ldots, x_{S_i}^{\left(i\right)} \} \subset {\cal X}$, so we can denote each $\mu_i = \sum_{s=1}^{S_i} q_s^{\left(i\right)} \delta_{x_s^{\left(i\right)}}$, where $S_i = |\text{supp}\left(\mu_i\right)|$ denotes the support size of $\mu_i$ and \textcolor{black}{$\{q_s^{\left(i\right)}\}$ is a sequence of nonnegative weights satisfying $\sum_{s=1}^{S_i} q_s^{\left(i\right)}=1$}.

In practice it is common to set \( k = p \) for the RWB \(\nu_{p, k}^{(\lambda)}\), similar to the classical WB. We \textcolor{black}{therefore} focus on the \(\nu_{p, p}^{(\lambda)}\). 
For convenience, let \(\nu_{p}^{(\lambda)}\) denote \(\nu_{p, p}^{(\lambda)}\), which is the solution to the optimization problem:
\begin{equation}
	\label{fv}
	\begin{aligned}
		&\underset{\nu \in \mathcal{W}^{\left(\lambda\right)}_{p}(\mathcal{X})}{\arg \min}f\left(\nu\right)
		,
		\\
		& f\left(\nu\right) = \sum_{i=1}^{n} w_i \left[ W_{p}^{\left(\lambda\right)}\left(\nu, \mu_i\right)\right]^{p} ,
	\end{aligned}
\end{equation}
where $\{w_i\}_{i=1}^{n}$ is a sequence of nonnegative weights satisfying $\sum_{i=1}^n w_i=1$.

The following theorem, which relies on \citet[Prop.~3]{MFT19} and \citet[Thm.~1]{GA27}, establishes a theoretical connection between Wasserstein barycenter problems and multi-marginal optimal transport. It further characterizes the support of WB and extends the corresponding theoretical framework to the RWB, as the RWB is a special case of the WB.
\begin{theorem}\label{thm.support} 
	According to Proposition~\ref{prop.exist}, equation~\eqref{fv} admits at least one solution, which is the RWB $\nu_{p}^{(\lambda)}$, and we have: 
		\begin{equation}
		\begin{aligned}
			\text{supp}
						\left(\nu_{p}^{(\lambda)}\right) \subseteq  
			\underset{y \in  \mathcal{X}}{\arg \min} \left\{ \sum_{i=1}^{n}w_ic^{\left(\lambda\right)} \left(y,x^{\left(i\right)}\right)^{p} | \forall x^{\left(i\right)}\in\text{supp}\left(\mu_i\right),\,i=1,\dots,n\right\}
		\end{aligned}
	\end{equation}
	Moreover, there exists an optimal solution $\nu^*$ s.t. 
	\begin{center}
		$ |\text{supp}\left(\nu^*\right)| \leq \sum_{i=1}^{n}S_i + 1 - n $.
	\end{center}
\end{theorem}
This theorem not only helps choose a suitable support for fixed-support algorithms in RWB computation, but also provides an upper bound on the support size for free-support algorithms.

\subsection{Free-support methods}
As noted above, we have that \textcolor{black}{$\text{supp}\left(\mu_i\right)$} and $\mu_i = \sum_{s=1}^{S_i} q_s^{\left(i\right)} \delta_{x_s^{\left(i\right)}}$, i.e., for $i=1,\dots,n$.
By Theorem~\ref{thm.support}, the support of $\nu_{p}^{(\lambda)}$ in \eqref{fv} is finite, so we may write supp($\nu_{p}^{(\lambda)}$) = $\{ y_1,\dots,y_R  \}$ and $\nu_{p}^{(\lambda)}=\sum_{r=1}^{R}p_r\delta_{y_r}$, where $y_r \in \mathcal{X}$ and \textcolor{black}{$\{p_r\}_{r=1}^{R}$ is a sequence of nonnegative weights satisfying $\sum_{r=1}^{R}p_r=1$}.

For $i=1,\dots,n$, let $\Pi^{i}$ denote the optimal transports plan $\Pi\left(\nu_{p}^{(\lambda)},\mu_i\right)$, i.e., 
$$\Pi^{i} = \sum_{r=1}^{R}\sum_{s=1}^{S_i}\pi_{rs}^{\left(i\right)}\delta\left(y_r,x_{s}^{\left(i\right)}\right), $$ 
where $\{\pi_{rs}^{\left(i\right)}\}$ minimizes $\sum_{r=1}^{R}\sum_{s=1}^{S_i}\pi_{rs}^{\left(i\right)}c^{(\lambda)}\left(y_r,x_{s}^{\left(i\right)}\right)^p$ and $\pi_{rs}^{\left(i\right)}$ denotes the weight of the coupling \(\Pi\left(\nu_{p}^{(\lambda)}, \mu_i\right)\) at \(\left(y_r, x_{s}^{(i)}\right)\). 

Given that $\pi^{i}=\left(\pi_{rs}^{\left(i\right)}\right)_{R \times S_i}$ and $\pi := \left(\pi^1,\dots,\pi^n\right) = \left(\pi_{rs}^{\left(i\right)}\right)_{n \times R \times S_i} $,
and we define a map $B_{\nu_{p}^{(\lambda)}} : \rm{supp}\left(\nu_{p}^{(\lambda)}\right) \rightarrow \mathcal{X}$ as follows:
\begin{equation}
	\label{eq.support} 
	b_r \stackrel{\triangle}{=} B_{\nu_{p}^{(\lambda)}}\left(y\right) :=  \underset{z \in  \mathcal{X}}{\arg \min} \sum_{i=1}^nw_i\sum_{s=1}^{S_i}\pi_{rs}^{\left(i\right)}c^{\left(\lambda\right)} \left(z,x^{\left(i\right)}\right)^{p},\quad
	r=1,\dots,R. 
\end{equation} 
In general, $w_i\left(i=1,\dots,n\right)$ is set to $\frac{1}{n}$ so we have  $$b_r =  \underset{z \in  \mathcal{X}}{\arg \min} \sum_{i=1}^n\sum_{s=1}^{S_i}\pi_{rs}^{\left(i\right)}c^{\left(\lambda\right)} \left(z,x^{\left(i\right)}\right)^{p},\quad
r=1,\dots,R.$$
Furthermore, we can define the new discrete measure as: $b\left(\nu_{p}^{(\lambda)}\right) := \sum_{r=1}^{R}p_r\delta_{b_r}$. \\
\textbf{Remark 4.2}
\textcolor{black}{Note that if \(d\) is continuous, then \(c^{(\lambda)} = \min(d, \lambda)\) is also continuous, which implies that \(\sum_{i=1}^nw_i\sum_{s=1}^{S_i}\pi_{rs}^{(i)}c^{(\lambda)} \left(z,x^{(i)}\right)^{p}\) is a continuous function on a compact space. This ensures the existence of \(b_r\) for \(r=1,\dots,R\).}

\begin{theorem}
	\label{reduce_f}
	For $\nu_{p}^{(\lambda)}$ and $b\left(\nu_{p}^{(\lambda)}\right)$, we have:
	\begin{center}
		$
		f\left(b\left(\nu_{p}^{(\lambda)}\right)\right) \leq f\left(\nu_{p}^{(\lambda)}\right),  
		$
	\end{center}
	where $f\left(\cdot\right)$ is defined in \eqref{fv}.
\end{theorem}
\begin{proof}
	Let $\Pi = \sum_{i=1}^{n}w_i\sum_{r=1}^{R}\sum_{s=1}^{S_i}{\pi}_{rs}^{\left(i\right)}\delta\left(y_r,x_{s}^{\left(i\right)}\right)$ denote the transport plan between $\nu_{p}^{(\lambda)}$ and $\{\mu_i\}_{i=1}^{n}$, and $\hat{\Pi} = \sum_{i=1}^{n}w_i\sum_{r=1}^{R}\sum_{s=1}^{S_i}\hat{\pi}_{rs}^{\left(i\right)}\delta\left(b_r,x_{s}^{\left(i\right)}\right)$ denote the transport plan between $b\left(\nu_{p}^{(\lambda)}\right)$ and $\{\mu_i\}_{i=1}^{n}$. Then we have:	
	$$
	\begin{aligned}
		f\left(b\left(\nu_{p}^{(\lambda)}\right)\right) &= 
		\sum_{i=1}^{n}w_i\left[W_{p}^{\left(\lambda\right)}\left(b\left(\nu_{p}^{(\lambda)}\right),\mu_i\right)\right]^p \\
		&= 
		\sum_{i=1}^{n}w_i\sum_{r=1}^{R}\sum_{s=1}^{S_i}\hat{\pi}_{rs}^{\left(i\right)}c^{\left(\lambda\right)}\left(b_r,x_s^{\left(i\right)}\right)^p \\
		&\leq
		\sum_{i=1}^{n}w_i\sum_{r=1}^{R}\sum_{s=1}^{S_i}\pi_{rs}^{\left(i\right)}c^{\left(\lambda\right)}\left(b_r,x_s^{\left(i\right)}\right)^p \\
		&= 
		\sum_{r=1}^{R}\sum_{i=1}^{n}w_i\sum_{s=1}^{S_i}\pi_{rs}^{\left(i\right)}c^{\left(\lambda\right)}\left(b_r,x_s^{\left(i\right)}\right)^p \\
		&=
		\sum_{r=1}^{R}\textcolor{black}{\underset{z \in \mathcal{X}}{\min}} \sum_{i=1}^{n}w_i\sum_{s=1}^{S_i}\pi_{rs}^{\left(i\right)}c^{\left(\lambda\right)}\left(z,x_s^{\left(i\right)}\right)^p \\
		&\leq
		\sum_{r=1}^{R}\sum_{i=1}^{n}w_i\sum_{s=1}^{S_i}\pi_{rs}^{\left(i\right)}c^{\left(\lambda\right)}\left(y_r,x_s^{\left(i\right)}\right)^p \\
		&=
		\sum_{i=1}^{n}w_iW_{p}^{\left(\lambda\right)}\left(\nu_{p}^{(\lambda)},\mu_i\right)^p \\
		&=
		f\left(\nu_{p}^{(\lambda)}\right) 
	\end{aligned}
	$$ 	\end{proof} 

Following the alternating update scheme, the support points are refined via equation~\eqref{eq.support}, followed by optimization of the associated probability weights. Optimizing the weights of a fixed support set guarantees the minimization of the sum of robust Wasserstein distances from the barycenter to the input distributions.
Denote the optimal weights as \(p^* = \left(p^*_1,\dots,p^*_R\right)\). The resulting barycenter satisfies
$$b^*\left(\nu_{p}^{(\lambda)}\right) := \sum_{r=1}^{R}p^*_r\delta_{b_r}.$$
In view of the optimality of \(p^*\), it holds that
$$f\left(b^*\left(\nu_{p}^{(\lambda)}\right)\right) \leq f\left(b\left(\nu_{p}^{(\lambda)}\right)\right).$$
We may then further optimize the support set of \(b^*\left(\nu_{p}^{(\lambda)}\right)\), i.e., \(b\left(b^*\left(\nu_{p}^{(\lambda)}\right)\right)\).

\textbf{Remark 4.3} Given an arbitrary discrete probability measure $\nu$,
we set $\nu_0 = \nu_{p}^{(\lambda)}, \nu_1 = b(\nu_0), \nu_k = b(b^*\left(\nu_{k-2}\right)),k=2,3,\dots$, then we can have $
f\left(\nu_{p}^{(\lambda)}\right) \geq f\left(\nu_1\right) \geq f\left(\nu_2\right) \geq \dots \geq 0$. Moreover, $\{f\left(\nu_k\right)\}$ converges by the monotone bounded convergence theorem.

According to Theorem~\ref{reduce_f}, replacing the measure $\nu_p^{(\lambda)}$ obtained by any fixed-support algorithm with $b(\nu_p^{(\lambda)})$ can reduce the sum of robust Wasserstein distances between the barycenter and input distributions.
Additionally, it is clear that the conclusions in this section apply not only to the robust Wasserstein space but also to other Wasserstein spaces. Algorithm~\ref{alg:free} provides the pseudocode of the free-support method.

\begin{algorithm}[h]
	\caption{Free-support method for RWB}
	\label{alg:free}
	\begin{algorithmic}[1] %
		\Require A sequence of discrete probability masses $\{\mu_i\}_{i=1}^n$ and their supports $\{\text{supp}(\mu_i)\}_{i=1}^n$, the cost \textcolor{black}{function $c^{\left(\lambda\right)}(\cdot,\cdot)^{p}$}, the original support of RWB supp$_0$  
		\While{not converge}
		\State Compute the optimal mass $\nu_{\text{mass}}$ by the fixed-support algorithm
		\State Update the support supp$_{\text{new}}$ by \eqref{eq.support}
		\EndWhile
		\Ensure $\nu_{\text{mass}}$ is the mass of RWB $\nu_{p}^{(\lambda)}$, and supp$_{\text{new}}$ is the support of  $\nu_{p}^{(\lambda)}$   
	\end{algorithmic}
\end{algorithm}

Furthermore, since the RWB is a special instance of the standard WB, various fixed-support WB solvers can be readily adapted within our flexible free-support framework for diverse practical scenarios.
Specifically, IBP \cite{IBP} and exact linear programming directly yield the optimal transport plan matrices $\Pi^{i}$.
Based on the derived optimal transport plans, we can not only update the support points but also directly compute the total transportation cost in Step 2 of Algorithm~\ref{alg:free}. In practical implementations, we terminate the iteration once either the total transportation cost or the support points converge, where the specific error threshold can be set according to practical numerical demands.

It is clear that equation~\eqref{eq.support} essentially solves the weighted $p$-th order barycenter problem with respect to $c^{(\lambda)}$. When $c^{(\lambda)}$ denotes the truncated Euclidean distance, support points can be updated via the block coordinate descent framework in \cite{BCD} or the subgradient descent method. For initialization, we can randomly sample support points from given input distributions.

\section{Numerical results} 

\subsection{Comparison of image processing}

The experiments in this section are thus restricted to low-dimensional image processing scenarios. The high-dimensional performance of the proposed method can be theoretically inferred from its sparsity and computational framework, and will be verified in future work with more sufficient computational resources.
{In this section, we compare the performance of the RWB, WB, and Wasserstein median (W-median) in \cite{YouShun} on simulated random elliptical images. The RWB and WB are computed via the IBP algorithm \cite{IBP} and its free-support variant derived earlier in this paper, while the W-median is implemented using Algorithm 1 in \cite{YouShun}.
Furthermore, as the free-support framework we mentioned earlier applies to RWB and WB (i.e., those minimizing the sum of \(W_{p}^p\)), and the W-median minimizes the sum of \(W_2\), we present only fixed-support results for the W-median.}
For any input image of $j \times \ell$ pixels, we can model it as a discrete probability measure supported on the grid $\{0,1, \ldots, j-1\}\times \{0,1, \ldots, \ell-1\}$. 
When computing the RWB and WB are computed by the free-support variant of the IBP algorithm, we take the support of the barycenter to consist of $R$ points in the space $[0, j-1] \times [0, l-1]$, where $R$ denotes the number of support points we selected for the barycenter. 
Given the sparsity of the support of the barycenter, it is reasonable to select \( R \) that is smaller than the total number of grid points.

In the experiments, we generate $n = 200$ random nested ellipse images of size $40\times 40$ pixels, each containing random perturbation points in the upper-right corner that serve as outliers. An example of such images is shown in Figure~\ref{elic}.
To avoid potential information loss (which requires a sufficiently large number of barycenter support points) while keeping $R$ from being excessively large, we test various values of \( R \) and  set \( R = 105 \). {Furthermore, we also test various values of \(\lambda\) to evaluate the corresponding results. We observe that values of \(\lambda\) that are neither too large nor too small—i.e., those striking a balance between information loss and robustness---yield optimal performance, consistent with the findings reported in \cite{ma2023theoretical}. Finally, we set \(\lambda = 5\) for the cost function \textcolor{black}{\(c^{(\lambda)}(\cdot,\cdot)\)}.}

{Additionally, in this experiment, we observed that the barycenter images computed using the ${W}_1$ distance are visually sharper than those obtained with ${W}_2$. Therefore, all the results of the RWB and WB presented below are computed with the ${W}_1$ distance.

From the results shown in Figure~\ref{elic_ans}, we draw the following conclusions:
\begin{enumerate}
	\item[(i)] The boundaries of the Wasserstein barycenter near the random perturbation points become blurred, accompanied by a noticeable shift toward the disturbance.
	\item[(ii)] Under both free-support and fixed-support settings, the RWB has sharper boundaries than the WB, with no obvious deviation toward the perturbation points.
	\item[(iii)] {Under the fixed-support setting, the W-median is less affected by the perturbations than the WB, but it is outperformed by the RWB.}
\end{enumerate}

This numerical exercise confirms that the RWB is more robust to random perturbation points when computed on random ellipse images.

\begin{figure}[h!t]
	\centering
	\includegraphics[width=.4\linewidth]{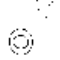}
	\caption{A random nested ellipse image with random outlier points in the upper-right corner.}
	\label{elic}
\end{figure}
\begin{figure}[h!t]
	\centering
	\setlength{\fboxsep}{2pt}  
	\setlength{\fboxrule}{0.2pt} 
	
	\begin{minipage}[b]{0.4\linewidth}
		\centering
		\fbox{\includegraphics[width=\linewidth]{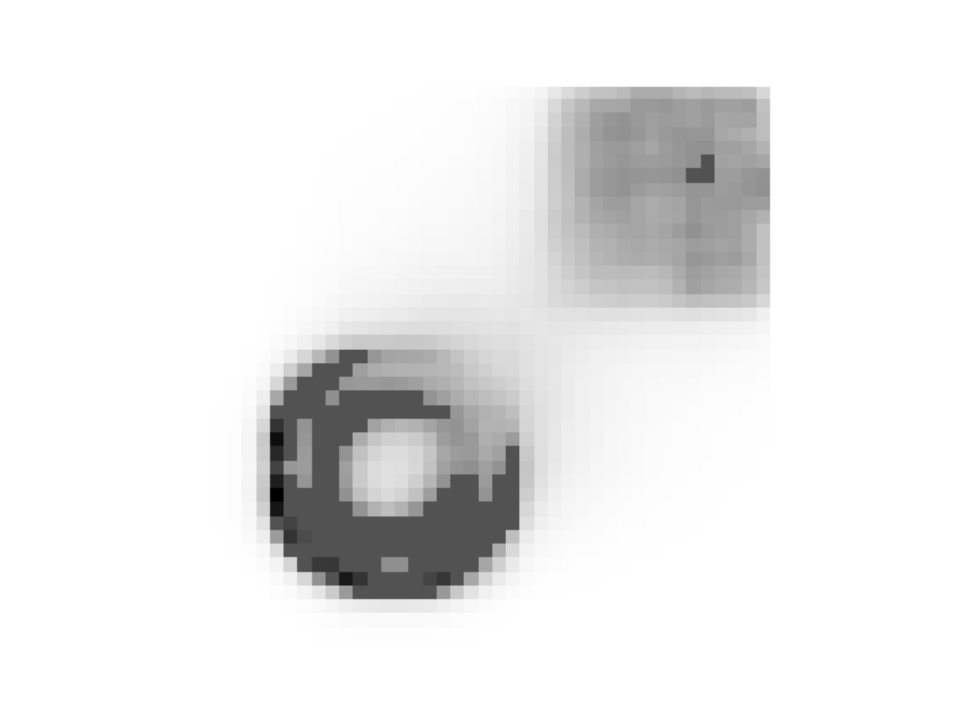}}
		\caption*{(a) Fixed-support WB}
	\end{minipage}
	\hfill %
	\begin{minipage}[b]{0.4\linewidth}
		\centering
		\fbox{\includegraphics[width=\linewidth]{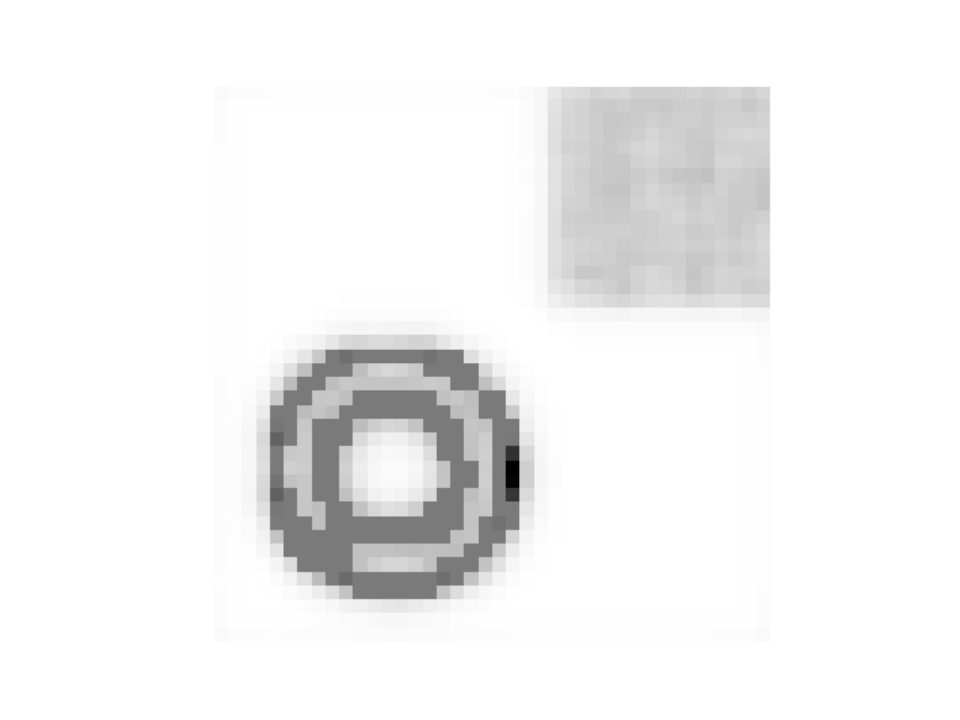}}
		\caption*{(b) Fixed-support RWB}
	\end{minipage}
	
	\begin{minipage}[b]{0.4\linewidth}
		\centering
		\fbox{\includegraphics[width=\linewidth]{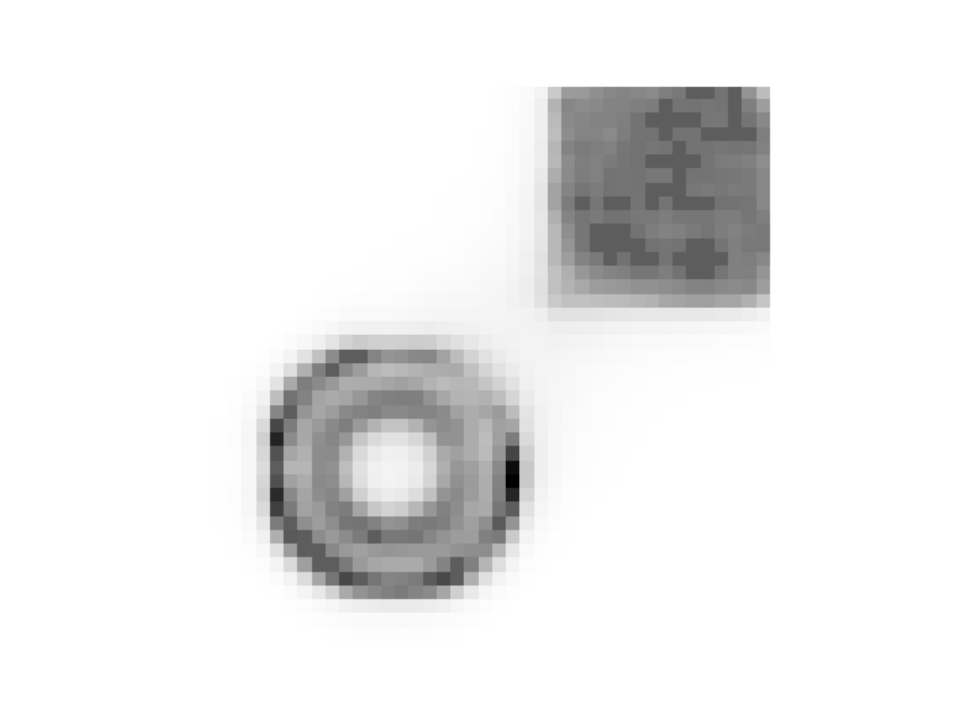}}
		\caption*{(c) Fixed-support W-median}
	\end{minipage}
	\vspace{1em} 
	
	\begin{minipage}[b]{0.4\linewidth}
		\centering
		\fbox{\includegraphics[width=\linewidth]{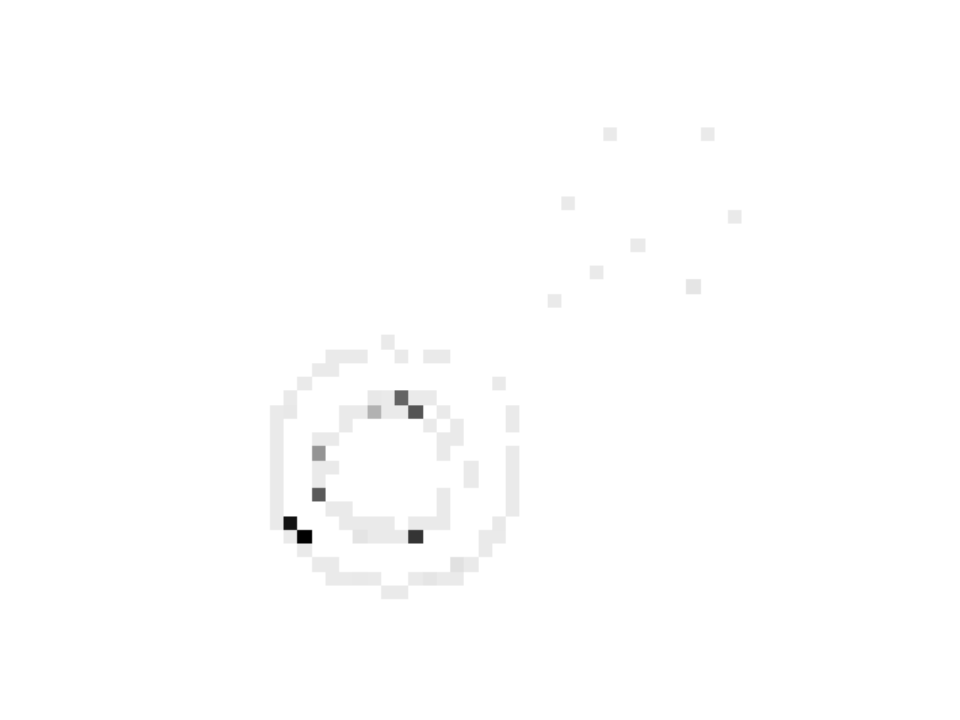}}
		\caption*{(d) Free-support WB}
	\end{minipage}
	\hfill
	\begin{minipage}[b]{0.4\linewidth}
		\centering
		\fbox{\includegraphics[width=\linewidth]{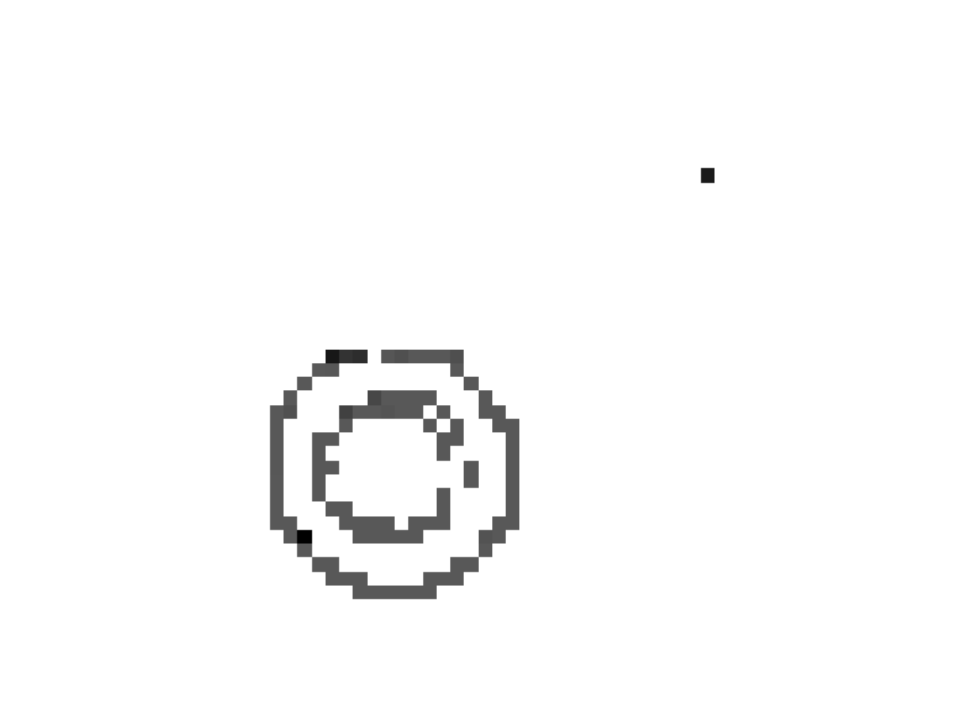}}
		\caption*{(e) Free-support RWB}
	\end{minipage}
	
	\caption{Barycenters of random nested ellipse images.}
	\label{elic_ans}
\end{figure}

\subsection{Comparison of numerical data}
\label{subsec:contamination}
In this section, we demonstrate the variations of the RWB, the WB and the W-median under different data contamination ratios through numerical simulations.  
In our experiment, we consider two source distributions: the first as the signal (Type 1 distributions) and the second as the contamination (Type 2 distributions): 
\begin{center}
	Type 1 : $N\left(\mu_0,\sigma_0\right),\mu_0 \sim U\left(-20,20\right),\sigma_0=1$\\
	Type 2 : $N\left(\mu_1,\sigma_1\right),\mu_1 \sim U\left(30,70\right),\sigma_1=1$
\end{center}
In this setting, the true Wasserstein barycenter of the signal is \( \mathcal{N}(0,1) \).

We first sample observations from the clean signal and contamination distributions to construct a $100\times 1000$ array, corresponding to $100$ independent datasets where all the samples within a single dataset share identical location parameters, and each dataset contains $1000$ samples.
{Then we employ $k$-means clustering to generate support points. After empirical tuning over different values, we fix the number of support points to 100, and derive 100 input distributions by calculating the frequency of each dataset with respect to these support points.}
Finally, we apply only the fixed-support algorithm to compute barycenters given the simplicity of this setting and the small number of support points.
In addition, the cost function \(c(\cdot,\cdot)\) is taken as the Euclidean distance,
and we test various truncation values \(\lambda \in \{10, 20, 30, 40, 50, 60, 70\}\) for the modified cost function \(c^{(\lambda)}\),
where 71.8 approximates the maximum value of the original cost matrix in this experiment. 
By analyzing the results under different values of \(\lambda\)  shown in Figure~\ref{diff_guass_bary} and Table~\ref{tab-lam}, we observe that the RWB achieves optimal performance at \(\lambda = 30\) (since it yields the smallest value of mean Wasserstein distance (under different contamination ratios) between the RWB and true WB).  

{Figure~\ref{gauss} illustrates the variations in Wasserstein distance from the WB, the W-median and the RWB to the true Wasserstein barycenter and across data contamination ratios ranging from 0\% to 25\% in 1\% increments.}
From the results shown in Figure~~\ref{gauss}, we draw the following conclusions:
\begin{enumerate}
	\item[(i)] The Wasserstein distances from the WB, W-median, and RWB to the true-WB monotonically increase as the data contamination ratio increases.
	\item[(ii)] The variation magnitude of the WB curve is slightly greater than that of the RWB curve, and the W-median exhibits relatively stable performance.
	\item[(iii)] Among the three methods, the RWB is closest to the true-WB and least affected by data contamination, followed by the W-median, with the WB performing the worst.	
\end{enumerate}
\begin{table}[!htbp]
	\centering
	\begin{tabular}{lccccccc}
		\hline 
		$\lambda$& 10& 20& 30& 40& 50& 60& 70 \\ \hline
		W.distance.mean& 10.822& 6.605& 6.590& 7.644& 8.779& 9.679& 10.450 \\ \hline 
	\end{tabular}
	\caption{Mean Wasserstein distance (under different contamination ratios) between the RWB and true WB across different values of $\lambda$.}
	\label{tab-lam}
\end{table}
\begin{figure}[h!t]
	\centering
	\includegraphics[width=.5\linewidth]{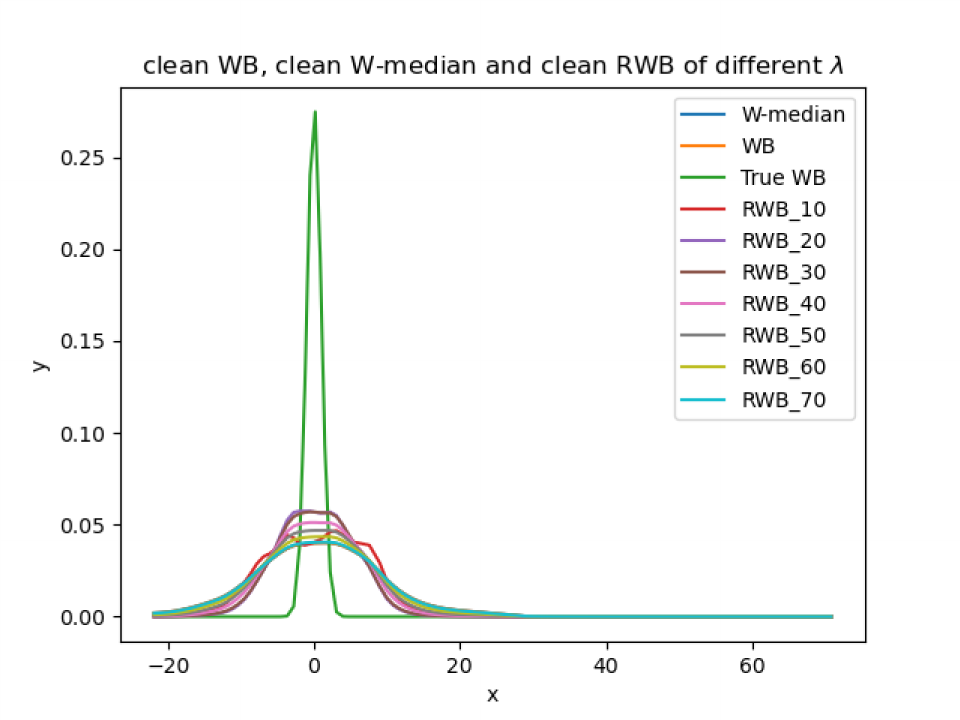}
	\caption{Distribution curves of barycenters.}
	\label{diff_guass_bary}
\end{figure}
\begin{figure}[h!t]
	\centering
	\includegraphics[width=.5\linewidth]{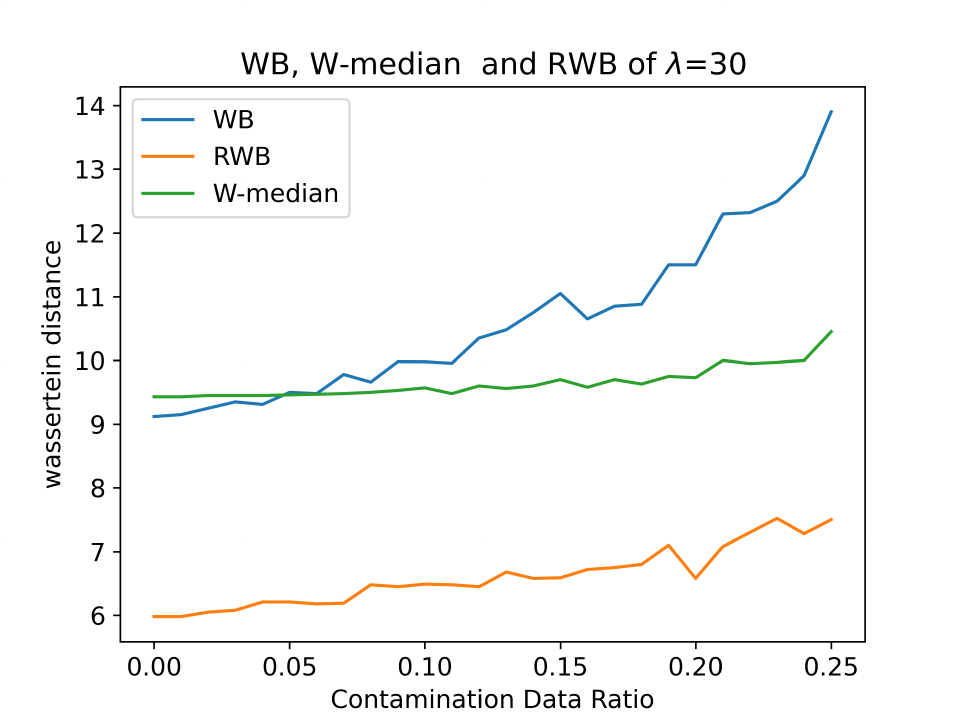}
	\caption{Wasserstein distances from the WB, RWB and W-median to the true WB.}
	\label{gauss}
\end{figure}

In addition to the contamination case, the RWB also exhibits greater robustness than both the WB and W-median in heavy-tailed cases. 
We thus examine data generated from the t-distribution $\text{t}(df,m,s)$, where $df$, $m$ and $s$ denote the degrees of freedom, location and scale, respectively. 
In this experiment, to ensure the existence of the second-order moment of the t-distribution, we set $\text{df}=3$, $s=1$, and $m \sim U(-50,50)$.
Similar to the contamination case, we test the truncation value \(\lambda\) within the range \(\lambda \in \{30,40,50,60,70,80,90,100,110\}\). Among these values, \(\lambda = 60\) yields the optimal result, as illustrated in Figure~\ref{t} and Figure~\ref{t_shape}. %
The results of the experiments demonstrate that the RWB outperforms both the WB and W-median for this heavy-tailed case.

\begin{figure}[htb]
	\centering  
	\includegraphics[width=.5\linewidth]{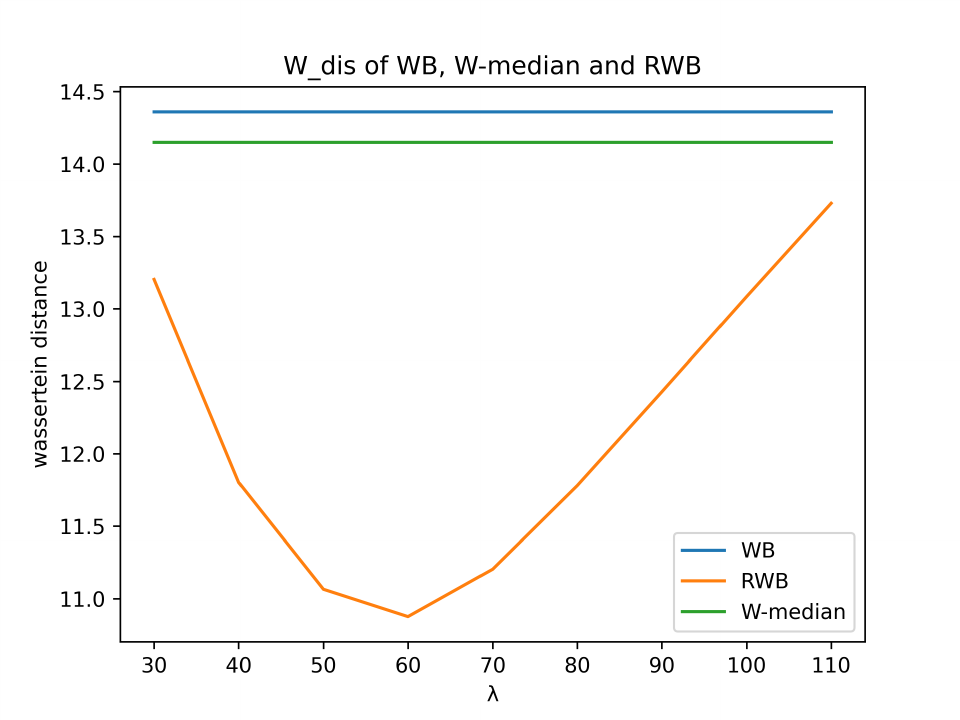}\vspace{12pt} 
	\caption{Wasserstein distances from the WB, RWB and W-median to the true WB in the heavy-tailed case.} 
	\label{t}
\end{figure} 
\begin{figure}[htb]
	\centering  
	\includegraphics[width=.45\linewidth]{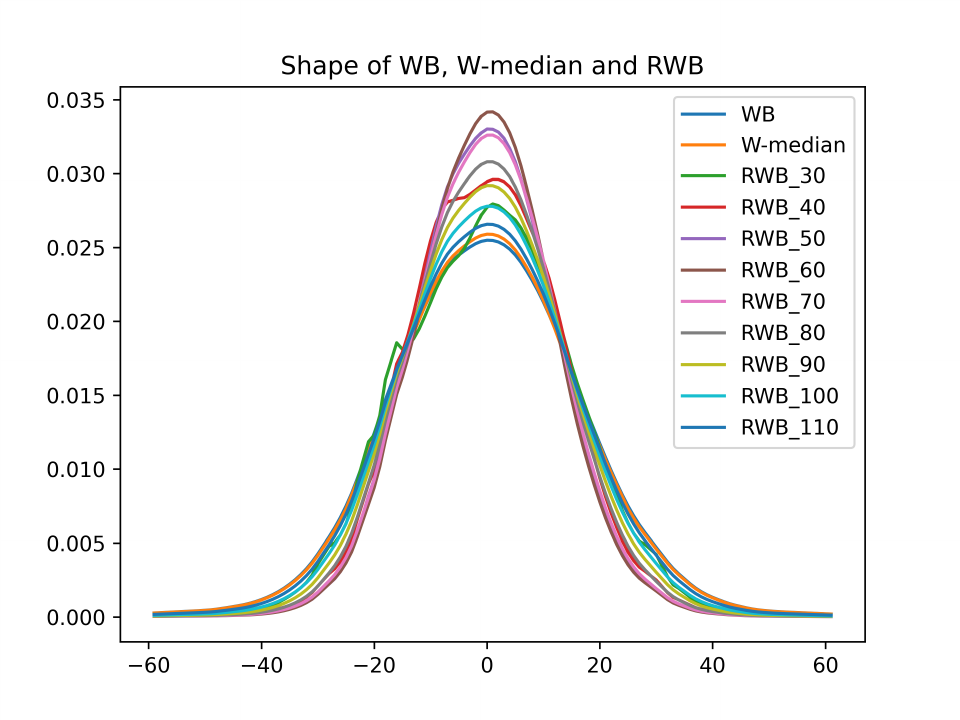}\vspace{12pt} 
	\caption{Distribution curves of barycenters in the heavy-tailed case.} 
	\label{t_shape}
\end{figure}

In summary, the simulation results presented above indicate that the RWB outperforms both the WB and the W-median in terms of robustness under both contamination and heavy-tailed cases.

\begin{remark}
\textcolor{black}{In practice, we should select a moderate $\lambda$---neither too large nor too small---to strike a balance between robustness and information preservation. For instance, $\lambda$ can be chosen based on the quantiles of the data, specifically as the distance between the quantile where the values begin to surge and the median of the gradually increasing values. Our extensive simulations show that for $\lambda$ within a wide range, the performance does not vary significantly; for example, as shown in Fig.~5, when $\lambda \in [40, 80]$, the value roughly fluctuates around the interval $[11.0, 12.0]$.}
\end{remark}

\subsection{Real data analysis}

In this section, we analyze the closing prices of the stock dataset in \cite{stockdata} and demonstrate the robustness of the RWB.
We calculate the daily log-returns \( r_t = \log(X_t / X_{t-1}) \) for the closing prices \( X_t \) of each stock, and we obtain the distribution of \( r_t \) for each stock via kernel density estimation (KDE).

In practice, the values of $r_t$ typically cluster around 0, yet this stock dataset contains several excessively large $r_t$ values, which are clearly abnormal.
Therefore, we use the Isolation Forest method proposed in \cite{isolation} to identify critical thresholds for outliers, partitioning \( r_t \) into two subsets: \textit{clean} and \textit{out}. These two subsets are then combined to form the complete set of total \( r_t \).
Furthermore, we apply k-means clustering to select 190 points from the \textit{clean} subset and 10 points from the \textit{out} subset as support points, which are used to characterize the distribution of \(r_t\) for each stock.
The total support set, denoted as support$_{total}$, is the union of support$_{clean}$ and support$_{out}$, comprising 200 support points in total. 
We then calculate the frequency distribution of the KDE for each stock over the \( \text{support}_{\text{clean}} \) points, which serves as the input distribution under the clean case.
Similarly, we compute the frequency distribution of the KDE for each stock over \( \text{support}_{\text{total}} \), which is used as the input distribution under the outlier case.

{For this setting, the barycenters derived from these input distributions enable us to observe the overall return distribution of stock closing prices corresponding to the market represented by the dataset.}

Additionally, we compare the barycenters calculated from the outlier-included case with those from the outlier-free (\textit{clean}) case to evaluate their robustness. Intuitively, a smaller deviation between the two scenarios indicates that the barycenter is less influenced by outliers, implying stronger robustness.

{As illustrated in Figure~\ref{stock_bary_dis}, the Wasserstein distance between RWB\(_{\text{total}}\) and WB\(_{\text{clean}}\) remains consistently smaller than those between WB\(_{\text{total}}\) and WB\(_{\text{clean}}\), as well as between W-median\(_{\text{total}}\) and WB\(_{\text{clean}}\), across various values of \(\lambda\), exhibiting minimal variation.
	Furthermore, W-median\(_{\text{total}}\) and WB\(_{\text{total}}\) perform almost identically, with the corresponding results being 0.8679 and 0.8691, respectively.}
Furthermore, Figure~\ref{stock_bary_shape} reveals that the shape of RWB\(_{\text{total}}\) most closely approximates that of WB\(_{\text{clean}}\) when \(\lambda = 0.28\). Consequently, we select \(\lambda = 0.28\) as the optimal parameter. The final barycenter results in Figure~\ref{stock_contrast} demonstrate that RWB\(_{\text{total}}\) is significantly closer to WB\(_{\text{clean}}\) than both WB\(_{\text{total}}\) and W-median\(_{\text{total}}\).
In summary, for this dataset containing substantial outliers, RWB proves to be significantly less susceptible to outliers than WB and W-median, thereby exhibiting superior robustness.
\begin{table}[!htbp]
	\centering
	\resizebox{\linewidth}{!}{
		\begin{tabular}{lccccccccccc}
			\hline 
			$\lambda$& 0.24& 0.26& 0.28& 0.30& 0.32& 0.34& 0.36& 0.38& 0.40& 0.42& 0.44 \\ \hline
			W.distance& 0.0616& 0.0629& 0.0641& 0.0652& 0.0661& 0.0671& 0.0679& 0.0686& 0.0694& 0.0700& 0.0706 \\ \hline 
		\end{tabular}
	}
	\caption{Wasserstein distance between RWB and the WB$_{\text{clean}}$ across different $\lambda$.}
	\label{tab-lam-wdis}
\end{table}
\begin{figure}[htb]
	\centering   
	\includegraphics[width=.5\linewidth]{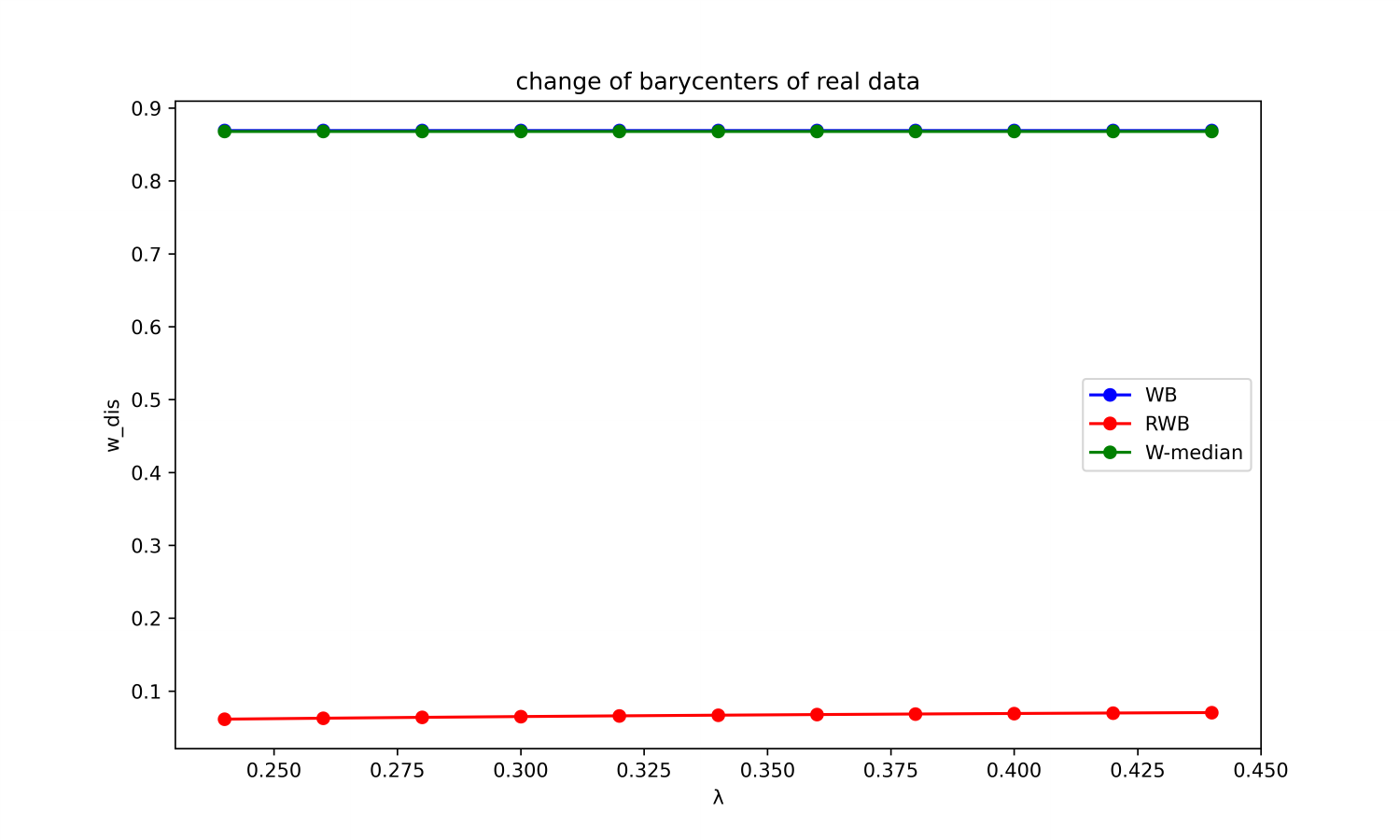}\vspace{12pt} 
	\caption{Wasserstein distance from WB\(_{\text{total}}\), RWB\(_{\text{total}}\), and W-median\(_{\text{total}}\) to WB\(_{\text{clean}}\) over different \(\lambda\).}  
	\label{stock_bary_dis}
\end{figure}
\begin{figure}[htb]
	\centering   
	\includegraphics[width=.5\linewidth]{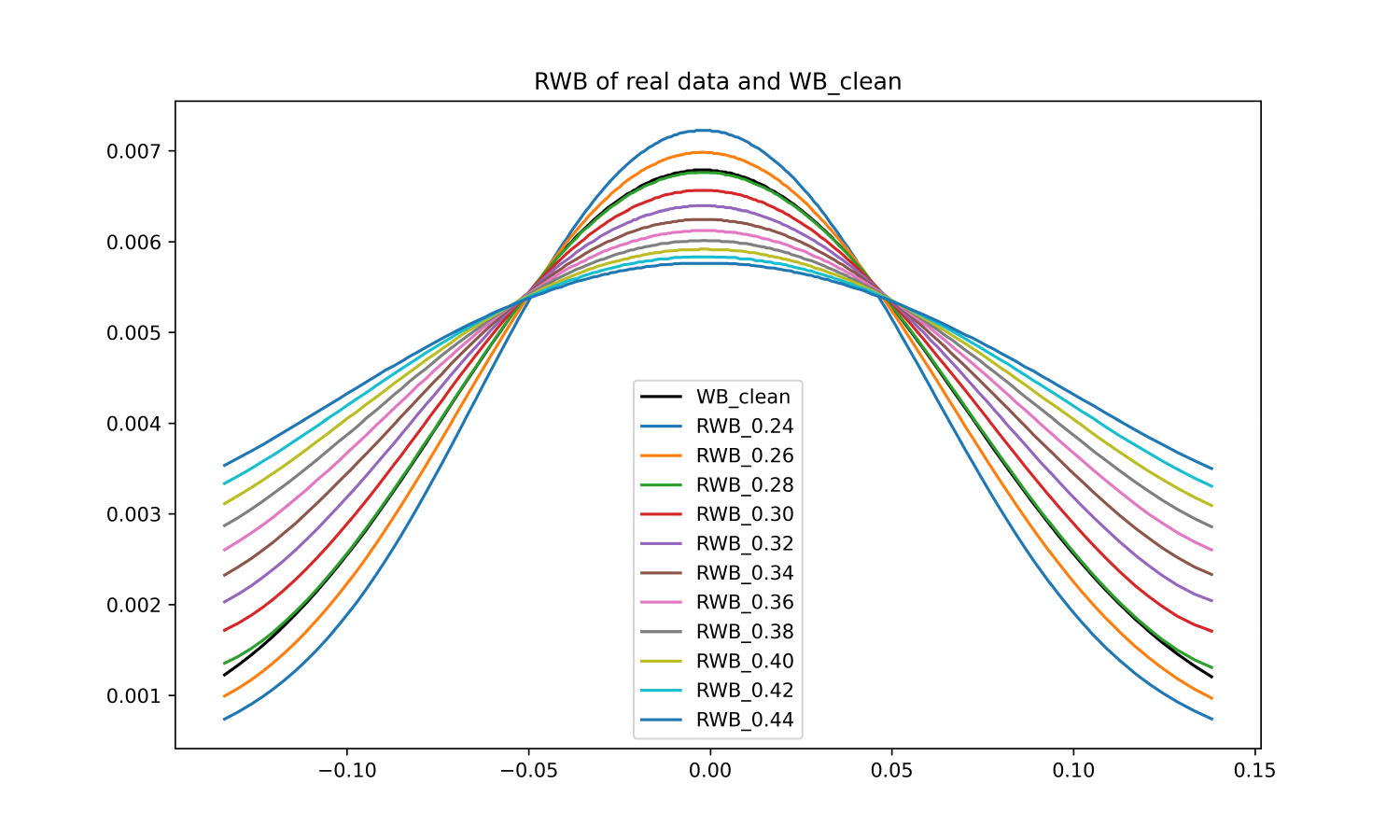}\vspace{12pt} 
	\caption{Curves of barycenters.}  
	\label{stock_bary_shape}
\end{figure}
\begin{figure}[htb]
	\centering     
	\includegraphics[width=.55\linewidth]{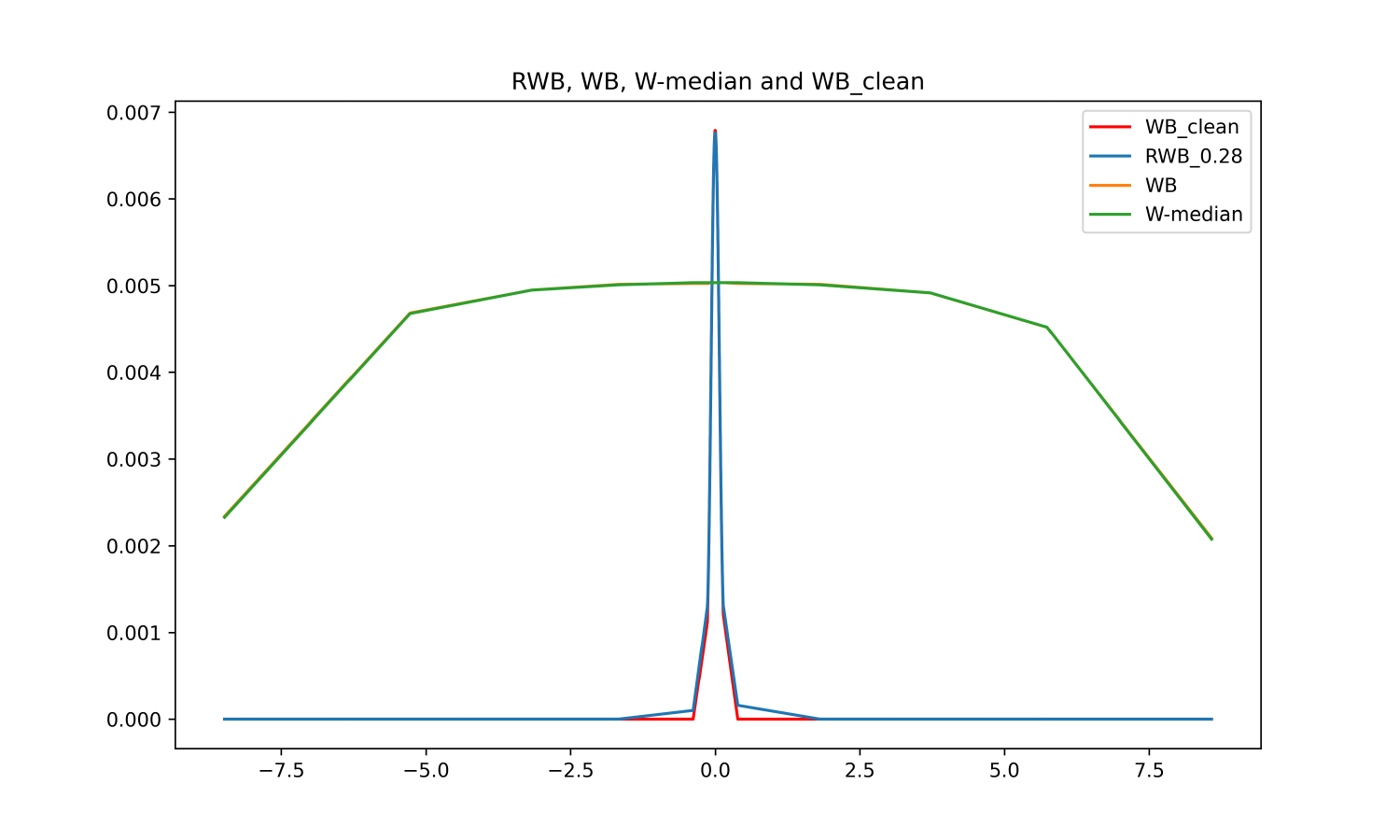}\vspace{12pt}
	\caption{The figure depicts the clean-case Wasserstein barycenter (WB\(_{\text{clean}}\)), WB$_{\text{total}}$, W-median$_{\text{total}}$ and RWB$_{\text{total}}$, where RWB is calculated with $\lambda=0.28$.}  
	\label{stock_contrast} 
\end{figure}

\section{Conclusions}

This paper introduces the robust Wasserstein barycenter to mitigate the instability of the classical Wasserstein barycenter in the presence of outliers. 
We provide theoretical guarantees for the RWB, including its existence and consistency, and propose an effective computational method. 
Extensive numerical analysis, comprising both simulated and real-world data, illustrates that the RWB is significantly more robust than the standard WB, making it a compelling alternative for practical use. 
For future research, one may consider an entropy-regularized version of the RWB, and explore its theoretical and empirical performance in the presence of outliers and high dimensionality. 
Additionally, one may combine the RWB with current cutting-edge deep learning models.








\section*{Conflict of Interest}
There is no conflict of interest to report.

\section*{Preprint Statement}
Research presented in this article was posted on a preprint server prior to publication in JUSTC. The corresponding preprint article can be found here: arXiv:2603.07563.

\section*{Funding Information}
Hang Liu's research was supported by the National Natural Science Foundation of China (No. 12401372) and University of Science and Technology of China (No. 2024ycjg13).

\section*{Biographies}
Zixiong Cheng is currently a master's student in statistics at the School of Management, University of Science and Technology of China, under the supervision of Associate Professor Hang Liu. His research interests focus mainly on optimal transport and statistics.

Hang Liu received his Ph.D. in Statistics from Lancaster University in 2021. He is currently an associate professor at the School of Management, University of Science and Technology of China. His research interests include optimal transport, nonparametric statistics, Riemannian geometry, and time series.

\end{document}